\documentclass[twoside, times, 11pt, journal, onecolumn]{IEEEtran}

\usepackage{epsfig}
\usepackage{graphicx}
\usepackage{amsmath}
\usepackage{amssymb}
\usepackage{float}
\usepackage{url}
\usepackage{verbatim}

\usepackage[bottom=0.8in]{geometry}

\newtheorem{theorem}{Theorem}

\newtheorem{proposition}{Proposition}
\newtheorem{corollary}{Corollary}

\newtheorem{definition}{Definition}

\newtheorem{remark}{Remark}
\newcommand{\expectation}{\ensuremath{\mathbb{E}}}
\newcommand{\prob}{\ensuremath{\mathbb{P}}}
\newcommand{\naturals}{\ensuremath{\mathbb{N}}}
\newcommand{\reals}{\ensuremath{\mathbb{R}}}

\begin{document}
\title{\Huge{On Reverse Pinsker Inequalities}}

\author{Igal Sason
\thanks{
I. Sason is  with the Department of Electrical Engineering, Technion--Israel
Institute of Technology, Haifa 32000, Israel (e-mail: sason@ee.technion.ac.il).
This work has been submitted in part to the {\em 2015 IEEE Information Theory
Workshop (ITW 2015)}, Jeju Island, Korea, October~11--15, 2015. The research
was supported by the Israeli Science Foundation (ISF), grant number 12/12.}}

\maketitle

\begin{abstract}
New upper bounds on the relative entropy are derived as a function of the total
variation distance. One bound refines an inequality by Verd\'{u} for general
probability measures. A second bound improves the tightness of an inequality
by Csisz\'{a}r and Talata for arbitrary probability measures that are defined
on a common finite set. The latter result is further extended, for probability
measures on a finite set, leading to an upper bound on the
R\'{e}nyi divergence of an arbitrary non-negative order (including $\infty$)
as a function of the total variation distance. Another lower bound by Verd\'{u}
on the total variation distance, expressed in terms of the distribution
of the relative information, is tightened and it is attained under some conditions.
The effect of these improvements is exemplified.
\end{abstract}


\thispagestyle{empty}

\section{Introduction}
Consider two probability measures $P$ and $Q$ defined on a common
measurable space $(\mathcal{A}, \mathcal{F})$.
The Csisz\'{a}r-Kemperman-Kullback-Pinsker inequality states that
\begin{equation}
D(P \| Q) \geq \frac{\log e}{2} \cdot |P-Q|^2
\label{eq: Pinsker's inequality}
\end{equation}
where
\begin{equation}
D(P \| Q) = \expectation_P\left[ \log \frac{\text{d}P}{\text{d}Q} \right]
= \int_{\mathcal{A}} \text{d}P \, \log \frac{\text{d}P}{\text{d}Q}
\label{eq: definition of the relative entropy}
\end{equation}
designates the relative entropy from $P$ to $Q$ (a.k.a. the
Kullback-Leibler divergence), and
\begin{equation}
|P-Q| = 2 \, \sup_{A \in \mathcal{F}} \bigl|P(A) - Q(A)\bigr|
\label{eq: TV distance}
\end{equation}
designates the total variation distance (or $L_1$ distance) between $P$ and $Q$.
One of the implications of inequality \eqref{eq: Pinsker's inequality} is that
convergence in relative entropy implies convergence in total variation distance.
The total variation distance is bounded $|P-Q| \leq 2$, in contrast to the relative
entropy.

Inequality \eqref{eq: Pinsker's inequality} is a.k.a.
Pinsker's inequality, although the analysis made by Pinsker \cite{Pinsker60}
leads to a significantly looser bound where $\frac{\log e}{2}$ on the
RHS of \eqref{eq: Pinsker's inequality} is replaced by
$\frac{\log e}{408}$ (see \cite[Eq.~(51)]{Verdu_ITA14}). Improved and generalized
versions of Pinsker's inequality have been studied in \cite{FedotovHT_IT03},
\cite{Gilardoni06}, \cite{Gilardoni10}, \cite{OrdentlichW_IT2005}, \cite{ReidW11},
\cite{Vajda_IT1970}.

For any $\varepsilon > 0$, there exists a pair of probability measures
$P$ and $Q$ such that $|P-Q| \leq \varepsilon$
while $D(P\|Q) = \infty$. Consequently, a reverse Pinsker inequality
which provides an upper bound on the relative entropy in terms of the total
variation distance does not exist in general. Nevertheless, under some
conditions, such inequalities hold \cite{CsiszarT_IT06},
\cite{Verdu_ITA14}, \cite{Verdu_book} (to be addressed later in this section).

If $P \ll Q$, the relative information in $a \in \mathcal{A}$ according to
$(P,Q)$ is defined to be
\begin{equation}
i_{P\|Q}(a) \triangleq \log \frac{\text{d}P}{\text{d}Q} \, (a).
\label{eq: relative information}
\end{equation}
From \eqref{eq: definition of the relative entropy}, the relative entropy
can be expressed in terms of the relative information as follows:
\begin{align}
D(P \| Q) = \expectation \bigl[ i_{P \| Q}(X) \bigr]
= \expectation \bigl[ i_{P \| Q}(Y) \, \exp\bigl(i_{P \| Q}(Y) \bigr) \bigr]
\label{eq: relative entropy expressed in terms of relative information}
\end{align}
where $X \sim P$ and $Y \sim Q$ (i.e., $X$ and $Y$ are distributed according
to $P$ and $Q$, respectively).
The total variation distance is also expressible
in terms of the relative information \cite{Verdu_ITA14}. If $Q \ll P$
\begin{equation}
|P-Q| = \expectation \Bigl[ \bigl| 1 - \exp\bigl(i_{P\|Q}(Y)\bigr) \bigr| \Bigr]
\label{eq: total variation distance expressed in terms of relative information1}
\end{equation}
and if, in addition, $P \ll Q$, then
\begin{equation}
|P-Q| = \expectation \Bigl[ \bigl| 1 - \exp\bigl(-i_{P\|Q}(X)\bigr) \bigr| \Bigr].
\label{eq: total variation distance expressed in terms of relative information2}
\end{equation}

Let
\begin{equation}
\beta_1^{-1} \triangleq \sup_{a \in \mathcal{A}} \frac{\text{d}P}{\text{d}Q} \, (a)
\label{eq: beta1}
\end{equation}
with the convention, implied by continuity, that $\beta_1 = 0$ if $i_{P\|Q}$
is unbounded from above.
With $\beta_1 \leq 1$, as it is defined in \eqref{eq: beta1}, the following
inequality holds (see \cite[Theorem~7]{Verdu_ITA14}):
\begin{align}
\frac{1}{2} |P-Q| & \geq \left( \frac{1-\beta_1}{\log \frac{1}{\beta_1}} \right) D(P \| Q).
\label{eq: Sergio's inequality - ITA14}
\end{align}
From \eqref{eq: Sergio's inequality - ITA14}, if the relative information
is bounded from above, a reverse Pinsker inequality holds. This inequality
has been recently used in the context of the optimal quantization of probability
measures when the distortion is either characterized by the total variation distance
or the relative entropy between the approximating and the original
probability measures \cite[Proposition~4]{BochererG_arXiv}.

Inequality~\eqref{eq: Sergio's inequality - ITA14} is refined in this work,
and the improvement that is obtained by this refinement is exemplified (see
Section~\ref{section: A Tightened Reverse Pinsker Inequality for General Probability Measures}).

In the special case where $P$ and $Q$ are defined
on a common discrete set (i.e., a finite or countable set) $\mathcal{A}$,
the relative entropy and total variation distance are simplified to
\begin{align*}
& D(P \| Q) = \sum_{a \in \mathcal{A}} P(a) \, \log \frac{P(a)}{Q(a)}, \\
& |P-Q| = \sum_{a \in \mathcal{A}} \bigl|P(a) - Q(a)\bigr| \triangleq |P-Q|_1.
\end{align*}

A restriction to probability measures on a finite set $\mathcal{A}$ has led
in \cite[p.~1012 and Lemma~6.3]{CsiszarT_IT06} to the following upper bound on the
relative entropy in terms of the total variation distance:
\begin{equation}
D(P \| Q) \leq \left(\frac{\log e}{Q_{\min}} \right)
\cdot |P-Q|^2,
\label{eq: looser bound by Csiszar and Talata}
\end{equation}
where $Q_{\min} \triangleq \min_{a \in \mathcal{A}} Q(a)$, suggesting a kind of a reverse
Pinsker inequality for probability measures on a finite set. A recent application
of this bound has been exemplified in \cite[Appendix~D]{KostinaV15} and \cite[Lemma~7]{TomamichelT_IT13}
for the analysis of the third-order asymptotics of the discrete memoryless channel with or
without cost constraints.

The present paper also considers generalized reverse Pinsker inequalities for
R\'{e}nyi divergences.
In the discrete setting, the R\'{e}nyi divergence of order $\alpha$ from $P$ to $Q$ is defined as
\begin{equation}
D_{\alpha}(P || Q) \triangleq \frac{1}{\alpha-1} \; \log \left( \sum_{a \in \mathcal{A}}
P^{\alpha}(a) \, Q^{1-\alpha}(a) \right), \quad \forall \, \alpha \in (0,1) \cup (1, \infty).
\label{eq: Renyi divergence}
\end{equation}
Recall that $D_{1}(P\|Q) \triangleq D(P\|Q)$ is defined to be the analytic extension
of $D_{\alpha}(P || Q)$ at $\alpha=1$ (if $D(P||Q) < \infty$, it can be verified with
L'H\^{o}pital's rule that $D(P||Q) = \lim_{\alpha \rightarrow 1^-} D_{\alpha}(P || Q)$).
The extreme cases of $\alpha = 0, \infty$ are defined as follows:
\begin{itemize}
\item If $\alpha = 0$ then $D_0(P||Q) = -\log Q(\text{Support}(P))$ where
$\text{Support}(P) = \{x \in \mathcal{X} \colon P(x)>0 \}$ denotes the support of $P$,
\item If $\alpha = +\infty$ then $D_{\infty}(P||Q) = \log \left(\text{ess sup} \, \frac{P}{Q}\right)$
where $\text{ess sup} \, f$ denotes the essential supremum of a function $f$.
\end{itemize}

Pinsker's inequality has been generalized by Gilardoni \cite{Gilardoni10}
for R\'{e}nyi divergences of order $\alpha \in (0,1]$
(see also \cite[Theorem~30]{ErvenH14}), and it gets the form
\begin{equation*}
D_{\alpha}(P \| Q) \geq \frac{\alpha \, \log e}{2} \cdot |P-Q|^2.
\end{equation*}
An improved bound, providing the best lower bound on the R\'{e}nyi divergence of order
$\alpha > 0$ in terms of the total variation distance, has been recently introduced in
\cite[Section~2]{Sason_arXiv15}.

Motivated by these findings, the analysis in this paper suggests an improvement over
the upper bound on the relative entropy in \eqref{eq: looser bound by Csiszar and Talata}
for probability measures defined on a common finite set.
The improved version of the bound in \eqref{eq: looser bound by Csiszar and Talata} is
further generalized to provide an upper bound on the R\'{e}nyi divergence of orders
$\alpha \in [0,\infty]$ in terms of the total variation distance.

Note that the issue addressed in this paper of deriving, under suitable conditions, upper
bounds on the relative entropy as a function of the total variation distance has some similarity
to the issue of deriving upper bounds on the difference between entropies as a function of the total
variation distance. Note also that in the special case where $Q$ is a Gaussian distribution
and $P$ is a probability distribution with the same covariance matrix, then
$D(P \| Q) = h(Q) - h(P)$ where $h(\cdot)$ denotes the differential entropy of a specified
distribution (see \cite[Eq.~(8.76)]{Cover_Thomas}). Bounds on the entropy difference in terms
of the total variation distance have been studied, e.g., in \cite[Theorem~17.3.3]{Cover_Thomas},
\cite{HoY_IT2010}, \cite{Prelov_PPI2007}, \cite{Prelov_PPI2008},
\cite{Sason_IT2013}, \cite[Section~1.7]{Verdu_book}, \cite{Zhang_IT2007}.

This paper is structured as follows:
Section~\ref{section: A Tightened Reverse Pinsker Inequality for General Probability Measures}
refers to \cite{Verdu_ITA14}, deriving a refined version of inequality \eqref{eq: Sergio's inequality - ITA14}
for general probability measures, and improving another lower bound on the total variation distance
which is expressed in terms of the distribution of the relative information.
Section~\ref{section: A New Reverse Pinsker Inequality for Probability Measures on a Finite Set}
derives a reverse Pinsker inequality for probability measures on a finite
set, improving inequality \eqref{eq: looser bound by Csiszar and Talata}
that follows from \cite[Lemma~6.3]{CsiszarT_IT06}.
Section~\ref{section: Extension to Renyi divergences} extends the analysis
to R\'{e}nyi divergences of arbitrary non-negative orders.
Section~\ref{section: The Exponential Decay of the Probability for a Non-Typical Sequence}
exemplifies the utility of a reverse Pinsker inequality in the context of typical sequences.

\section{A Refined Reverse Pinsker Inequality for General Probability Measures}
\label{section: A Tightened Reverse Pinsker Inequality for General Probability Measures}
The present section derives a reverse Pinsker inequality for general probability measures,
suggesting a refined version of \cite[Theorem~7]{Verdu_ITA14}. The utility of this new
inequality is exemplified. This section also provides a lower bound
on the total variation distance which is based on the distribution of the relative information;
the latter inequality is based on a modification of the proof of \cite[Theorem~8]{Verdu_ITA14},
and it has the advantage of being tight for a double-parameter family of probability measures
which are defined on an arbitrary set of 2~elements.

\subsection{Main Result and Proof}
Inequality \eqref{eq: Sergio's inequality - ITA14} provides an upper bound on the relative
entropy $D(P \|Q)$ as a function of the total variation distance when $P \ll Q$, and the
relative information $i_{P\|Q}$ is bounded from above (this implies that $\beta_1$
in \eqref{eq: beta1} is positive). The following theorem tightens this upper bound.

\vspace*{0.1cm}
\begin{theorem}
Let $P$ and $Q$ be probability measures on a measurable space $(\mathcal{A}, \mathcal{F})$,
$P \ll Q$, and let $\beta_1, \beta_2 \in [0,1]$ be given by
\begin{align}
\beta_1^{-1} \triangleq \sup_{a \in \mathcal{A}} \frac{\text{d}P}{\text{d}Q} (a), \quad
\beta_2 \triangleq \inf_{a \in \mathcal{A}} \frac{\text{d}P}{\text{d}Q} (a).
\label{eq: beta1 and beta2}
\end{align}
Then, the following inequality holds:
\begin{align}
D(P\|Q) \leq \frac{1}{2} \left( \frac{\log \frac{1}{\beta_1}}{1-\beta_1} - \beta_2 \log e \right) |P-Q|.
\label{eq: tightened version of Sergio's inequality}
\end{align}
\label{theorem: tightened version of Sergio's inequality}
\end{theorem}

\begin{proof}
Let $X \sim P$, $Y \sim Q$, and
\begin{align}
\mathcal{B} \triangleq \bigl\{ a \in \mathcal{A} \colon i_{P\|Q}(a) > 0 \bigr\}.
\label{eq: set B}
\end{align}
From \eqref{eq: relative entropy expressed in terms of relative information}, the relative entropy is equal to
\begin{align}
D(P\|Q) & = \int_{\mathcal{A}} \text{d}Q \, \exp\bigl(i_{P \| Q}\bigr) \, i_{P \| Q} \nonumber \\[0.1cm]
& = \int_{\mathcal{B}} \text{d}Q \, \exp\bigl(i_{P \| Q}\bigr) \, i_{P \| Q} +
\int_{\mathcal{A} \setminus \mathcal{B}} \text{d}Q \, \exp\bigl(i_{P \| Q}\bigr) \, i_{P \| Q}.
\label{eq: relative entropy - 1st step}
\end{align}
In the following, the two integrals on the RHS of \eqref{eq: relative entropy - 1st step} are upper bounded.
The upper bound on the first integral on the RHS of \eqref{eq: relative entropy - 1st step}
is based on the proof of \cite[Theorem~7]{Verdu_ITA14};
it is provided in the following for completeness, and with more details in order to clarify
the way that this bound is refined here. Let $z(a) \triangleq \exp(i_{P\|Q}(a))$ for
$a \in \mathcal{A}$. By assumption $1 < z(a) \leq \frac{1}{\beta_1}$
for all $a \in \mathcal{B}$. The function $f(z) = \frac{z \log(z)}{z-1}$ is monotonic increasing over the
interval $(1, \infty)$ since we have $(z-1)^2 f'(z) = (z-1) \log e - \log z > 0$ for $z>1$. Consequently, we have
\begin{align}
\frac{z(a) \log z(a)}{z(a)-1} \leq \frac{\log \frac{1}{\beta_1}}{1-\beta_1}, \quad \forall \, a \in \mathcal{B}
\label{eq: monotonicity of the function f - Sergio's proof}
\end{align}
and
\begin{align}
& \int_{\mathcal{B}} \text{d}Q \, \exp\bigl(i_{P \| Q}\bigr) \, i_{P \| Q} \nonumber \\[0.1cm]
& \stackrel{(\text{a})}{\leq} \frac{\log \frac{1}{\beta_1}}{1-\beta_1} \, \int_{\mathcal{B}} \text{d}Q \,
\left(\exp(i_{P\|Q})-1 \right) \nonumber \\[0.1cm]
& \stackrel{(\text{b})}{=} \frac{\log \frac{1}{\beta_1}}{1-\beta_1} \, \int_{\mathcal{A}} \text{d}Q(a) \,
\left(1-\exp(i_{P\|Q}(a)) \right)^{-} \nonumber \\[0.1cm]
& \stackrel{(\text{c})}{=} \left(\frac{\log \frac{1}{\beta_1}}{1-\beta_1} \right) \, \expectation
\Bigl[ \left(1-\exp(i_{P\|Q}(Y)) \right)^{-} \Bigr] \nonumber \\[0.1cm]
& \stackrel{(\text{d})}{=} \left( \frac{\log \frac{1}{\beta_1}}{2(1-\beta_1)} \right) |P-Q|
\label{eq: upper bound on first term of the relative entropy}
\end{align}
where inequality~(a) follows from \eqref{eq: monotonicity of the function f - Sergio's proof},
equality~(b) is due to \eqref{eq: set B} and the definition $(a)^{-} \triangleq -a \, 1\{a<0\}$,
equality~(c) holds since $Y \sim Q$, and equality~(d) follows from \cite[Eq.~(14)]{Verdu_ITA14}.

At this point, we deviate from the analysis in \cite{Verdu_ITA14} where the second integral on the
RHS of \eqref{eq: relative entropy - 1st step} has been upper bounded by zero (since
$i_{P\|Q}(a) \leq 0$ for all $a \in \mathcal{A} \setminus \mathcal{B}$). If $\beta_2>0$,
we provide in the following a strictly negative upper bound on this integral. Since
$P \ll Q$, we have
\begin{align}
& \int_{\mathcal{A} \setminus \mathcal{B}} \text{d}Q \, \exp\bigl(i_{P \| Q}\bigr)
\, i_{P \| Q} \nonumber \\[0.15cm]
& \stackrel{(\text{a})}{=} \int_{\{a \in \mathcal{A} \colon i_{P\|Q}(a)<0\}} \text{d}Q(a)
\; \; \frac{\text{d}P}{\text{d}Q}\, (a) \; \; i_{P \| Q}(a) \nonumber \\[0.15cm]
& \stackrel{(\text{b})}{\leq} \beta_2 \; \int_{\{a \in \mathcal{A} \colon i_{P\|Q}(a)<0\}}
\text{d}Q(a) \; i_{P \| Q}(a) \nonumber \\[0.15cm]
& \stackrel{(\text{c})}{\leq} \beta_2 \, \log e \; \int_{\{a \in \mathcal{A} \colon i_{P\|Q}(a)<0\}}
\text{d}Q(a) \; \Bigl(\exp\bigl(i_{P \| Q}(a)\bigr)-1\Bigr) \nonumber \\[0.15cm]
& \stackrel{(\text{d})}{=} -\beta_2 \, \log e \; \int_{\mathcal{A} \setminus \mathcal{B}}
\text{d}Q(a) \; \Bigl(1-\exp\bigl(i_{P \| Q}(a)\bigr)\Bigr) \nonumber \\[0.15cm]
& \stackrel{(\text{e})}{=} -\beta_2 \, \log e \; \int_{\mathcal{A}} \text{d}Q(a)
\; \Bigl(1-\exp\bigl(i_{P \| Q}(a)\bigr)\Bigr)^{+} \nonumber \\[0.15cm]
& \stackrel{(\text{f})}{=} -\beta_2 \log e \cdot \expectation
\Bigl[ \Bigl(1-\exp(i_{P\|Q}(Y)) \Bigr)^{+} \Bigr] \nonumber \\[0.15cm]
& \stackrel{(\text{g})}{=} -\frac{\beta_2 \log e }{2} \cdot |P-Q|
\label{eq: upper bound on second term of the relative entropy}
\end{align}
where equality~(a) holds due to \eqref{eq: relative information}, \eqref{eq: set B} and since
the integrand is zero if $i_{P\|Q}=0$, inequality~(b) follows from the definition of $\beta_2$ in
\eqref{eq: beta1 and beta2} and since $i_{P\|Q}$ is negative over the domain of integration, inequality~(c)
holds since the inequality $x \leq \log e \, \bigl(\exp(x)-1\bigr)$ is satisfied for all $x \in \reals$,
equalities~(d) and~(e) follow from the definition of the set $\mathcal{B}$ in \eqref{eq: set B},
equality~(f) holds since $Y \sim Q$, and equality~(g) follows  from \cite[Eq.~(15)]{Verdu_ITA14}.

Inequality~\eqref{eq: tightened version of Sergio's inequality} finally follows by combining
\eqref{eq: relative entropy - 1st step}, \eqref{eq: upper bound on first term of the relative entropy}
and \eqref{eq: upper bound on second term of the relative entropy}.
\end{proof}

\subsection{Example for the Refined Inequality in Theorem~\ref{theorem: tightened version of Sergio's inequality}}
We exemplify in the following the improvement obtained by \eqref{eq: tightened version of Sergio's inequality},
in comparison to~\eqref{eq: Sergio's inequality - ITA14}, due to the introduction of the additional parameter
$\beta_2$ in \eqref{eq: beta1 and beta2}. Note that when $\beta_2$ is replaced by zero (i.e., no information
on the infimum of $\frac{\text{d}P}{\text{d}Q}$ is available or $\beta_2=0$), inequalities
\eqref{eq: Sergio's inequality - ITA14} and \eqref{eq: tightened version of Sergio's inequality} coincide.

Let $P$ and $Q$ be two probability measures, defined on a set $\mathcal{A}$, where $P \ll Q$
and assume that
\begin{align}
1-\eta \leq \frac{\text{d}P}{\text{d}Q} \, (a) \, \leq 1+\eta, \quad \forall \, a \in \mathcal{A}
\label{eq: condition for 1st example}
\end{align}
for a fixed constant $\eta \in (0,1)$.

In \eqref{eq: tightened version of Sergio's inequality}, one can replace $\beta_1$ and $\beta_2$ with
lower bounds on these constants. From \eqref{eq: beta1 and beta2}, we have
$\beta_1 \geq \frac{1}{1+\eta}$ and $\beta_2 \geq 1-\eta$, and it follows from
\eqref{eq: tightened version of Sergio's inequality} that
\begin{align}
D(P\|Q) & \leq \frac{1}{2} \left( \frac{(1+\eta) \, \log(1+\eta)}{\eta} - (1-\eta) \log e \right)
\, |P-Q| \nonumber \\
& \leq \frac{1}{2} \Bigl( (1+\eta) \log e - (1-\eta) \log e \Bigr) \, |P-Q| \nonumber \\
& = \bigl( \eta \log e \bigr) |P-Q|.
\label{eq: relative entropy - bound1}
\end{align}
From \eqref{eq: condition for 1st example}
$$ \bigl| \exp\bigl(i_{P\|Q}(a)\bigr) - 1 \bigr| \leq \eta, \quad \forall \, a \in \mathcal{A}$$
so, from \eqref{eq: total variation distance expressed in terms of relative information1},
the total variation distance satisfies (recall that $Y \sim Q$)
$$|P-Q| =  \expectation\Bigl[ \bigl| \exp\bigl(i_{P\|Q}(Y)\bigr) - 1 \bigr| \Bigr]\leq \eta.$$
Combining the last inequality with \eqref{eq: relative entropy - bound1} gives that
\begin{align}
D(P\|Q) \leq \eta^2 \, \log e, \quad \forall \, \eta \in (0,1).
\label{eq: upper bound on relative entropy from strengthened version of Sergio's inequality}
\end{align}
For comparison, it follows from \eqref{eq: Sergio's inequality - ITA14} (see \cite[Theorem~7]{Verdu_ITA14})
that
\begin{align}
D(P\|Q) & \leq \frac{\log \frac{1}{\beta_1}}{2(1-\beta_1)} \cdot |P-Q| \nonumber \\[0.1cm]
& \leq \frac{(1+\eta) \, \log(1+\eta)}{2 \eta} \cdot |P-Q| \nonumber \\
& \leq \frac{1}{2} \, (1+\eta) \log(1+\eta) \nonumber \\
& \leq \left(\frac{\log e}{2} \right) \eta (1+\eta).
\label{eq: upper bound on relative entropy from Sergio's inequality}
\end{align}
Let $\eta \approx 0$. The upper bound on the relative entropy in
\eqref{eq: upper bound on relative entropy from Sergio's inequality} scales like $\eta$ whereas the tightened
bound in \eqref{eq: upper bound on relative entropy from strengthened version of Sergio's inequality} scales
like $\eta^2$. The scaling in \eqref{eq: upper bound on relative entropy from strengthened version of Sergio's inequality}
is correct, as it follows from Pinsker's inequality. For example, consider the probability measures defined on a
two-element set $\mathcal{A} = \{a,b\}$ with
$$P(a) = Q(b) = \frac{1}{2} - \frac{\eta}{4}, \quad P(b) = Q(a) = \frac{1}{2} + \frac{\eta}{4}.$$
Condition \eqref{eq: condition for 1st example} is satisfied for $\eta \approx 0$, and Pinsker's inequality yields that
\begin{align}
D(P\|Q) \geq \left( \frac{\log e}{2} \right) \eta^2
\label{eq: lower bound on relative entropy from Pinsker's inequality}
\end{align}
so the ratio of the upper and lower bounds in
\eqref{eq: upper bound on relative entropy from strengthened version of Sergio's inequality}
and \eqref{eq: lower bound on relative entropy from Pinsker's inequality} is~2,
and both provide the true quadratic scaling in $\eta$ whereas the weaker upper bound in
\eqref{eq: upper bound on relative entropy from Sergio's inequality} scales linearly
in $\eta$ for $\eta \approx 0$.

\subsection{Another Lower Bound on the Total Variation Distance}
The following lower bound on the total variation distance is based on the distribution of the
relative information, and it improves the lower bounds in \cite[Eq.~(2.3.18)]{Pinsker60},
\cite[Lemma~7]{SteinbergS} and \cite[Theorem~8]{Verdu_ITA14} by modifying
the proof of the latter theorem in \cite{Verdu_ITA14}. Besides of improving the tightness
of the bound, the motivation for the derivation of the following lower bound is that it
is achieved under some conditions.

\begin{theorem}
If $P$ and $Q$ are mutually absolutely continuous probability measures, then
\begin{equation}
|P-Q| \geq \sup_{\eta > 0} \Bigl\{ \bigl(1-\exp(-\eta)\bigr)
\; \Bigl( \prob\bigl[i_{P\|Q}(X) \geq \eta \bigr]
+ \exp(\eta) \; \prob\bigl[i_{P\|Q}(X) \leq -\eta \bigr] \Bigr) \Bigr\}
\label{eq: improved lower bound on total variation distance - Sergio}
\end{equation}
where $X \sim P$. This lower bound is attained if $P$ and $Q$ are
probability measures on a 2-element set $\mathcal{A} = \{a,b\}$ where,
for an arbitrary $\eta>0$,
\begin{equation}
P(a) = \frac{\exp(\eta)-1}{2 \sinh(\eta)}, \quad Q(a) = \frac{1-\exp(-\eta)}{2\sinh(\eta)}.
\label{eq: condition for a tight lower bound}
\end{equation}
\label{theorem: improved lower bound on total variation distance based on the distribution of the relative information}
\end{theorem}
\begin{proof}
Since $P \ll \gg Q$, we have
\begin{align*}
|P-Q| & = \expectation\bigl[ \bigl| 1 - \exp\bigl(-i_{P\|Q}(X)\bigr) \bigr| \bigr] \nonumber \\[0.1cm]
& \geq \expectation\bigl[ \bigl| 1 - \exp\bigl(-i_{P\|Q}(X)\bigr) \bigr|
\; 1\bigl\{\bigl| i_{P\|Q}(X) \bigr| \geq \eta \bigr\} \bigr], \quad \forall \, \eta > 0
\end{align*}
where $1\{\cdot\}$ is the indicator function of the specified event
(it is equal to~1 if the event occurs, and it is zero otherwise).
At this point we deviate from the proof of \cite[Theorem~8]{Verdu_ITA14}, and write
\begin{align}
|P-Q| & \geq \expectation\bigl[ \bigl| 1 - \exp\bigl(-i_{P\|Q}(X)\bigr) \bigr|
\; 1\bigl\{i_{P\|Q}(X) \geq \eta \bigr\} \bigr] \nonumber \\[0.1cm]
& \hspace*{0.4cm} + \expectation\bigl[ \bigl| 1 - \exp\bigl(-i_{P\|Q}(X)\bigr) \bigr|
\; 1\bigl\{i_{P\|Q}(X) \leq -\eta \bigr\} \bigr] \nonumber \\[0.1cm]
& \stackrel{\text{(a)}}{\geq} \bigl(1 - \exp(-\eta) \bigr) \; \expectation\bigl[1\bigl\{i_{P\|Q}(X)
\geq \eta \bigr\} \bigr] + \bigl( \exp(\eta)-1 \bigr) \;
\expectation\bigl[1\bigl\{i_{P\|Q}(X) \leq -\eta \bigr\} \bigr] \nonumber \\[0.1cm]
& = \bigl(1 - \exp(-\eta) \bigr) \; \Bigl( \prob\bigl[i_{P\|Q}(X) \geq \eta \bigr]
+ \exp(\eta) \; \prob\bigl[i_{P\|Q}(X) \leq -\eta \bigr] \Bigr)
\label{eq: improved lower bound on total variation distance before taking the supremum}
\end{align}
where step~(a) follows from the inequality $\bigl|1-\exp(-z)\bigr| \geq 1-\exp(-\eta)$
if $z \geq \eta$, and $\bigl|1-\exp(-z)\bigr| \geq \exp(\eta)-1$ if $z \leq -\eta$.
Taking the supremum on the right-hand side of
\eqref{eq: improved lower bound on total variation distance before taking the supremum},
w.r.t. the free parameter $\eta>0$, gives the lower bound on $|P-Q|$ in
\eqref{eq: improved lower bound on total variation distance - Sergio}.

The condition \eqref{eq: condition for a tight lower bound} for the tightness of the lower bound in
\eqref{eq: improved lower bound on total variation distance - Sergio} follows from the fact that, for
an arbitrary $\eta>0$, we have
$\log\left(\frac{P(a)}{Q(a)}\right) = \eta$ and $\log\left(\frac{1-P(a)}{1-Q(a)}\right) = -\eta$.
This yields that the inequalities in the derivation of the lower bound
\eqref{eq: improved lower bound on total variation distance - Sergio}
turn to be satisfied with equalities.
\end{proof}

\begin{remark}
One can further tighten the lower bound in
\eqref{eq: improved lower bound on total variation distance - Sergio} by writing,
for arbitrary $\eta_1, \eta_2 > 0$,
\begin{align*}
|P-Q| & \geq \expectation\bigl[ \bigl| 1 - \exp\bigl(-i_{P\|Q}(X)\bigr) \bigr|
\; 1\bigl\{i_{P\|Q}(X) \geq \eta_1 \bigr\} \bigr] \nonumber \\[0.1cm]
& \hspace*{0.4cm} + \expectation\bigl[ \bigl| 1 - \exp\bigl(-i_{P\|Q}(X)\bigr) \bigr|
\; 1\bigl\{i_{P\|Q}(X) \leq -\eta_2 \bigr\} \bigr]
\end{align*}
and proceeding similarly to
\eqref{eq: improved lower bound on total variation distance before taking the supremum}
to get the following lower bound on the total variation distance:
\begin{align*}
|P-Q| \geq \sup_{\eta_1, \eta_2 > 0} \biggl\{ & \bigl(1-\exp(-\eta_1)\bigr)
\; \biggl( \prob\bigl[i_{P\|Q}(X) \geq \eta_1 \bigr] \\[0.1cm]
& + \left(\frac{\exp(\eta_2)-1}{1-\exp(-\eta_1)} \right) \; \prob\bigl[i_{P\|Q}(X)
\leq -\eta_2 \bigr] \biggr) \biggr\}.
\end{align*}
This lower bound is achieved if $P$ and $Q$ are
probability measures on a 2-element set $\mathcal{A} = \{a,b\}$ where,
for an arbitrary $\eta_1, \eta_2>0$,
\begin{align}
P(a) = \frac{\exp(\eta_1)-\exp(\eta_1-\eta_2)}{\exp(\eta_1)-\exp(-\eta_2)}, \quad
Q(a) = \frac{1-\exp(-\eta_2)}{{\exp(\eta_1)-\exp(-\eta_2)}}
\label{eq: generalized condition for a tight lower bound}
\end{align}
which implies that $\log\left(\frac{P(a)}{Q(a)}\right) = \eta_1$ and
$\log\left(\frac{1-P(a)}{1-Q(a)}\right) = -\eta_2$.
Condition \eqref{eq: generalized condition for a tight lower bound} is specialized
to the condition in \eqref{eq: condition for a tight lower bound} when $\eta_1=\eta_2=\eta>0$.
\end{remark}

\section{A Reverse Pinsker Inequality for Probability Measures on a Finite Set}
\label{section: A New Reverse Pinsker Inequality for Probability Measures on a Finite Set}
The present section introduces a strengthened version of
inequality~\eqref{eq: looser bound by Csiszar and Talata} (see
Theorem~\ref{theorem: upper bound on KL divergence in terms of TV distance for probability measures on a finite set})
as a reverse Pinsker inequality for probability measures on a finite set,
followed by a discussion and an example.

\subsection{Main Result and Proof}
\begin{theorem}
Let $P$ and $Q$ be probability measures defined on a common finite set $\mathcal{A}$, and assume
that $Q$ is strictly positive on $\mathcal{A}$. Then, the following inequality holds:
\begin{equation}
D(P \| Q) \leq \log \left(1 + \frac{|P-Q|^2}{2 Q_{\min}} \right) - \frac{\beta_2 \log e}{2} \cdot |P-Q|^2
\label{eq: refined upper bound on KL divergence in terms of TV distance for probability measures on a finite set}
\end{equation}
where
\begin{align}
Q_{\min} \triangleq \min_{a \in \mathcal{A}} Q(a) > 0, \quad \beta_2 \triangleq \min_{a \in \mathcal{A}} \frac{P(a)}{Q(a)} \in [0,1].
\label{eq: Q_min and beta2}
\end{align}
\label{theorem: upper bound on KL divergence in terms of TV distance for probability measures on a finite set}
\end{theorem}

\begin{remark}
The upper bound on the relative entropy in
Theorem~\ref{theorem: upper bound on KL divergence in terms of TV distance for probability measures on a finite set}
improves the bound in \eqref{eq: looser bound by Csiszar and Talata}.
The improvement in \eqref{eq: refined upper bound on KL divergence in terms of TV distance for probability measures on a finite set}
is demonstrated as follows: let $V \triangleq |P-Q|$, then the RHS
of \eqref{eq: refined upper bound on KL divergence in terms of TV distance for probability measures on a finite set} satisfies
\begin{align*}
\log\left(1+\frac{V^2}{2Q_{\min}}\right) - \frac{\beta_2 \log e}{2} \cdot V^2
\leq \log\left(1+\frac{V^2}{2Q_{\min}}\right)
\leq \frac{V^2 \, \log e}{2Q_{\min}}
\leq \frac{V^2 \, \log e}{Q_{\min}}.
\end{align*}
Hence, the upper bound on $D(P\|Q)$ in
Theorem~\ref{theorem: upper bound on KL divergence in terms of TV distance for probability measures on a finite set}
can be loosened to \eqref{eq: looser bound by Csiszar and Talata}.
\end{remark}

\vspace*{0.1cm}
\begin{proof}
Theorem~\ref{theorem: upper bound on KL divergence in terms of TV distance for probability measures on a finite set}
is proved by obtaining upper and lower bounds on the $\chi^2$-divergence from $P$ to~$Q$
\begin{align}
\chi^2(P,Q) \triangleq \sum_{a \in \mathcal{A}} \frac{(P(a)-Q(a))^2}{Q(a)}
= \sum_{a \in \mathcal{A}} \frac{P(a)^2}{Q(a)} - 1.
\label{eq: chi-squared divergence}
\end{align}
A lower bound follows by invoking Jensen's inequality:
\begin{align}
\chi^2(P,Q)
& = \sum_{a \in \mathcal{A}} \frac{P(a)^2}{Q(a)} - 1 \nonumber \\
& = \sum_{a \in \mathcal{A}} P(a) \, \exp\left( \log \frac{P(a)}{Q(a)} \right) - 1 \nonumber \\
& \geq \exp\left(\sum_{a \in \mathcal{A}} P(a) \, \log \frac{P(a)}{Q(a)} \right) - 1 \nonumber \\[0.1cm]
& = \exp\bigl( D(P \| Q) \bigr) - 1.
\label{eq: lower bound on the chi-square divergence in terms of the KL divergence}
\end{align}
A refined version of \eqref{eq: lower bound on the chi-square divergence in terms of the KL divergence}
is derived in the following. The starting point of its derivation relies on a refined version of Jensen's
inequality from \cite[Theorem~1]{Dragomir06}, which enables to get the inequality
\begin{align}
\min_{a \in \mathcal{A}} \frac{P(a)}{Q(a)} \cdot D(Q||P)
\leq \log \bigl( 1 + \chi^2(P,Q) \bigr) - D(P||Q)
\leq \max_{a \in \mathcal{A}} \frac{P(a)}{Q(a)} \cdot D(Q||P).
\label{eq: relation between relative entropy, dual of relative entropy and chi-squared divergence}
\end{align}
Inequality~\eqref{eq: relation between relative entropy, dual of relative entropy and chi-squared divergence} is
proved in the appendix.
From the LHS of \eqref{eq: relation between relative entropy, dual of relative entropy and chi-squared divergence}
and the definition of $\beta_2$ in \eqref{eq: Q_min and beta2}, we have
\begin{align}
\chi^2(P,Q) & \geq \exp \Bigl( D(P\|Q) + \beta_2 \, D(Q \|P) \Bigr) - 1 \nonumber \\[0.1cm]
& \geq \exp \left( D(P \|Q) + \frac{\beta_2 \, \log e}{2} \cdot |P-Q|^2 \right) - 1
\label{eq: refined lower bound on the chi^2 divergence}
\end{align}
where the last inequality relies on Pinsker's lower bound on $D(Q\|P)$.
Inequality~\eqref{eq: refined lower bound on the chi^2 divergence} refines
the lower bound in \eqref{eq: lower bound on the chi-square divergence in terms of the KL divergence}
since $\beta_2 \in [0,1]$, and it coincides with
\eqref{eq: lower bound on the chi-square divergence in terms of the KL divergence}
in the worst case where $\beta_2=0$.

An upper bound on $\chi^2(P, Q)$ is derived as follows:
\begin{align}
\chi^2(P,Q) & = \sum_{a \in \mathcal{A}} \frac{(P(a)-Q(a))^2}{Q(a)} \nonumber \\[0.1cm]
& \leq \frac{\sum_{a \in \mathcal{A}} \bigl(P(a)-Q(a)\bigr)^2}{ \min_{a \in \mathcal{A}} Q(a) } \nonumber \\[0.1cm]
& \leq \frac{\max_{a \in \mathcal{A}} |P(a)-Q(a)| \, \sum_{a \in \mathcal{A}}
\bigl|P(a)-Q(a)\bigr|}{ \min_{a \in \mathcal{A}} Q(a) } \nonumber \\[0.1cm]
& = \frac{|P-Q| \; \max_{a \in \mathcal{A}} |P(a)-Q(a)|}{Q_{\min}}
\label{eq: first intermediate step for the derivation of the upper bound on the chi-square divergence}
\end{align}
and, from \eqref{eq: TV distance},
\begin{align}
|P-Q| \geq 2 \max_{a \in \mathcal{A}} |P(a)-Q(a)|
\label{eq: second intermediate step for the derivation of the upper bound on the chi-square divergence}
\end{align}
since, for every $a \in \mathcal{A}$, the 1-element set $\{a\}$ is included in the $\sigma$-algebra $\mathcal{F}$.
Combining \eqref{eq: first intermediate step for the derivation of the upper bound on the chi-square divergence} and
\eqref{eq: second intermediate step for the derivation of the upper bound on the chi-square divergence} gives that
\begin{align}
\chi^2(P,Q) \leq \frac{|P-Q|^2}{2Q_{\min}}.
\label{eq: upper bound on the chi-square divergence}
\end{align}
Inequality~\eqref{eq: refined upper bound on KL divergence in terms of TV distance for probability measures on a finite set}
finally follows from the bounds on the $\chi^2$-divergence in
\eqref{eq: refined lower bound on the chi^2 divergence} and \eqref{eq: upper bound on the chi-square divergence}.
\end{proof}

\vspace*{0.1cm}
\begin{corollary}
Under the same setting as in Theorem~\ref{theorem: upper bound on KL divergence in terms of TV distance for probability measures on a finite set},
we have
\begin{equation}
D(P \| Q) \leq \log \left(1 + \frac{|P-Q|^2}{2 Q_{\min}} \right).
\label{eq: upper bound on KL divergence in terms of TV distance for probability measures on a finite set}
\end{equation}
\label{corollary: upper bound on KL divergence in terms of TV distance for probability measures on a finite set}
\end{corollary}
\begin{proof}
This inequality follows from
\eqref{eq: refined upper bound on KL divergence in terms of TV distance for probability measures on a finite set}
and since $\beta_2 \geq 0$.
\end{proof}

\subsection{Discussion}
In the following, we discuss
Theorem~\ref{theorem: upper bound on KL divergence in terms of TV distance for probability measures on a finite set}
and its proof, and link it to some related results.

\vspace*{0.1cm}
\begin{remark}
The combination of \eqref{eq: lower bound on the chi-square divergence in terms of the KL divergence}
with the second line of \eqref{eq: first intermediate step for the derivation of the upper bound on the chi-square divergence},
without further loosening the upper bound on the $\chi^2$-divergence as is done in the third line of
\eqref{eq: first intermediate step for the derivation of the upper bound on the chi-square divergence} and
inequality~\eqref{eq: second intermediate step for the derivation of the upper bound on the chi-square divergence},
gives the following tighter upper bound on the relative entropy in terms of the Euclidean norm $|P-Q|_2$:
\begin{equation}
D(P \| Q) \leq \log \left(1 + \frac{|P-Q|_2^2}{Q_{\min}} \right).
\label{eq: upper bound on KL divergence in terms of the Euclidean norm of P-Q for probability measures on a finite set}
\end{equation}
This improves the upper bound on the relative entropy in the proofs of Property~4 of \cite[Lemma~7]{TomamichelT_IT13}
and \cite[Appendix~D]{KostinaV15}:
$$ D(P \| Q) \leq \frac{|P-Q|_2^2 \, \log e}{Q_{\min}}.$$
Furthermore, avoiding the use of Jensen's inequality in
\eqref{eq: lower bound on the chi-square divergence in terms of the KL divergence}, gives the equality
(see \cite[Eq.~(6)]{ErvenH14})
\begin{align}
\chi^2(P,Q) = \exp\bigl(D_2(P\|Q)\bigr)-1
\label{eq: relation between the chi-square and quadratic Renyi divergences}
\end{align}
whose combination with the second line of
\eqref{eq: first intermediate step for the derivation of the upper bound on the chi-square divergence} gives
\begin{equation}
D_2(P \| Q) \leq \log \left(1 + \frac{|P-Q|_2^2}{Q_{\min}} \right).
\label{eq: upper bound on quadratic Renyi divergence in terms of the Euclidean norm of P-Q for probability measures on a finite set}
\end{equation}
Inequality~\eqref{eq: upper bound on quadratic Renyi divergence in terms of the Euclidean norm of P-Q for probability measures on a finite set}
improves the tightness of
inequality~\eqref{eq: upper bound on KL divergence in terms of the Euclidean norm of P-Q for probability measures on a finite set}.
Note that \eqref{eq: upper bound on quadratic Renyi divergence in terms of the Euclidean norm of P-Q for probability measures on a finite set}
is satisfied with equality when $Q$ is an equiprobable distribution over a finite set.
\end{remark}

\vspace*{0.1cm}
\begin{remark}
Inequality \eqref{eq: lower bound on the chi-square divergence in terms of the KL divergence} improves the
lower bound on the $\chi^2$-divergence in \cite[Lemma~6.3]{CsiszarT_IT06} which states that
$\chi^2(P,Q) \geq D(P \| Q)$; this improvement also follows from \cite[Eqs.~(6), (7)]{ErvenH14}.
\end{remark}

\vspace*{0.1cm}
\begin{remark}
The upper bound on the relative entropy in
\eqref{eq: refined upper bound on KL divergence in terms of TV distance for probability measures on a finite set}
involves the parameter $\beta_2 \in [0,1]$ as defined in \eqref{eq: Q_min and beta2}. A non-trivial lower bound
on $\beta_2$ can be used in conjunction with
\eqref{eq: refined upper bound on KL divergence in terms of TV distance for probability measures on a finite set}
for improving the upper bound in
Corollary~\ref{corollary: upper bound on KL divergence in terms of TV distance for probability measures on a finite set}.
We derive in the following a lower bound on $\beta_2$ for a given probability measure $Q$ and a given total variation
distance $|P-Q|$, which can be used in conjunction with
\eqref{eq: refined upper bound on KL divergence in terms of TV distance for probability measures on a finite set},
to get an upper bound on the relative entropy $D(P\|Q)$. We have
$$\beta_2 = \min_{a \in \mathcal{A}} \frac{P(a)}{Q(a)} \geq \frac{P_{\min}}{Q_{\max}}$$
where $$P_{\min} \triangleq \min_{a \in \mathcal{A}} P(a), \quad Q_{\max} \triangleq \max_{a \in \mathcal{A}} Q(a).$$
Note that, if $|P-Q| < Q_{\min}$ then $P_{\min} \geq Q_{\min} - |P-Q| > 0$. Let $(x)^+ \triangleq \max \bigl\{ x, 0 \bigr\}$, then
\begin{align}
\beta_2 \geq \frac{\bigl(Q_{\min} - |P-Q|\bigr)^{+}}{Q_{\max}}.
\label{eq: lower bound on beta_2}
\end{align}
\end{remark}

\vspace*{0.1cm}
\begin{remark}
In an attempt to extend the concept of proof of
Theorem~\ref{theorem: upper bound on KL divergence in terms of TV distance for probability measures on a finite set}
to general probability measures, we have
\begin{align}
\chi^2(P,Q) & = \int_{\mathcal{A}} \left( \frac{\text{d}P}{\text{d}Q} - 1 \right)^2 \, \text{d}Q \nonumber \\[0.1cm]
& = \expectation \Bigl[ \bigl( \exp\bigl(i_{P \|Q}(Y) \bigr) - 1 \bigr)^2 \Bigr]  \quad \quad (Y \sim Q) \nonumber \\[0.1cm]
& \leq \sup_{a \in \mathcal{A}} \, \bigl| \exp\bigl(i_{P \|Q}(a) \bigr) - 1 \bigr| \cdot
\expectation \Bigl[ \bigl| \exp\bigl(i_{P \|Q}(Y) \bigr) - 1 \bigr| \Bigr]  \nonumber \\[0.1cm]
& \stackrel{(\text{a})}{=} \sup_{a \in \mathcal{A}} \, \bigl| \exp\bigl(i_{P \|Q}(a) \bigr) - 1 \bigr| \cdot |P-Q| \nonumber \\[0.1cm]
& = \sup_{a \in \mathcal{A}} \, \Bigl| \frac{\text{d}P}{\text{d}Q} \, (a) - 1 \Bigr| \cdot |P-Q|
\label{eq: 1st step in upper bounding the chi-squared divergence}
\end{align}
where equality~(a) holds due to~\eqref{eq: total variation distance expressed in terms of relative information1}.
Let $\beta_1, \beta_2 \in [0,1]$ be defined as in Theorem~\ref{theorem: tightened version of Sergio's inequality}
(see \eqref{eq: beta1 and beta2}). Since we have $\beta_2 \leq \frac{\text{d}P}{\text{d}Q} \, (a) \leq \beta_1^{-1}$
for all $a \in \mathcal{A}$ then
\begin{align}
\sup_{a \in \mathcal{A}} \, \Bigl| \frac{\text{d}P}{\text{d}Q} \, (a) - 1 \Bigr| \leq \max\bigl\{ \beta_1^{-1}-1, 1-\beta_2 \bigr\}.
\label{eq: 2nd step in upper bounding the chi-squared divergence}
\end{align}
A combination of \eqref{eq: 1st step in upper bounding the chi-squared divergence} and
\eqref{eq: 2nd step in upper bounding the chi-squared divergence} leads to the following
upper bound on the $\chi^2$-divergence:
\begin{align}
\chi^2(P,Q) \leq \max\bigl\{ \beta_1^{-1}-1, 1-\beta_2 \bigr\} \cdot |P-Q|.
\label{eq: upper bound on the chi-squared divergence for general probability measures}
\end{align}
A combination of \eqref{eq: relation between the chi-square and quadratic Renyi divergences} (see \cite[Eq.~(6)]{ErvenH14})
and \eqref{eq: upper bound on the chi-squared divergence for general probability measures} gives
\begin{align}
D_2(P \| Q) \leq \log \Bigl( 1 + \max\bigl\{ \beta_1^{-1}-1, 1-\beta_2 \bigr\} \cdot |P-Q| \Bigr)
\label{eq: upper bound on the quadratic Renyi divergence for general probability measures}
\end{align}
and since the R\'{e}nyi divergence is monotonic non-decreasing in its order (see, e.g., \cite[Theorem~3]{ErvenH14})
and $D(P\|Q) = D_1(P\|Q)$, it follows that
\begin{align}
D(P \| Q) \leq \log \Bigl( 1 + \max\bigl\{ \beta_1^{-1}-1, 1-\beta_2 \bigr\} \cdot |P-Q| \Bigr).
\label{eq: loosened upper bound on the relative entropy for general probability measures}
\end{align}
A comparison of the upper bound on the relative entropy in
\eqref{eq: loosened upper bound on the relative entropy for general probability measures}
and the bound of Theorem~\ref{theorem: tightened version of Sergio's inequality} in
\eqref{eq: tightened version of Sergio's inequality} yields that the latter bound is superior.
Hence, the extension of the concept of proof of
Theorem~\ref{theorem: upper bound on KL divergence in terms of TV distance for probability measures on a finite set}
to general probability measures does not improve the bound in Theorem~\ref{theorem: tightened version of Sergio's inequality}.
\label{remark: extension of the proof to general probability measures}
\end{remark}

\vspace*{0.1cm}
\begin{remark}
The second inequality in \eqref{eq: refined lower bound on the chi^2 divergence} relies on
Pinsker's inequality as a lower bound on $D(Q\|P)$. This lower bound can be slightly improved
by invoking higher-order Pinsker's-type inequalities (see \cite[Section~5]{Gilardoni10} and references
therein). In \cite[Section~6]{Gilardoni10}, Gilardoni derived a lower bound on the relative
entropy which is tight for both large and small total variation distances. Hence, the
second inequality in \eqref{eq: refined lower bound on the chi^2 divergence} can instead rely
on the inequality (see \cite[Eq.~(2)]{Gilardoni10}):
\begin{align*}
D(Q\|P) \geq -\log \left(1 - \frac{|P-Q|}{2} \right) - \left(1 - \frac{|P-Q|}{2}\right) \,
\log \left(1+\frac{|P-Q|}{2}\right).
\end{align*}
Note that although the latter lower bound on the relative entropy is tight for both large and small
total variation distances, it is not uniformly tighter than Pinsker's inequality. For this reason and
for the simplicity of the bound, we rely on Pinsker's inequality in the second inequality of
\eqref{eq: refined lower bound on the chi^2 divergence}.
An exact parametrization of the minimum of the relative entropy in terms of
the total variation distance was introduced in \cite[Theorem~1]{FedotovHT_IT03}, expressed in
terms of hyperbolic functions; the bound, however, is not expressed in closed form in terms of
the total variation distance.
\end{remark}

\vspace*{0.1cm}
\begin{remark}
A related problem to
Theorem~\ref{theorem: upper bound on KL divergence in terms of TV distance for probability measures on a finite set}
has been recently studied in \cite{BerendHK_IT14}. Consider an arbitrary probability measure $Q$, and
an arbitrary $\varepsilon \in [0,2]$. The problem studied in \cite{BerendHK_IT14}
is the characterization of $D^*(\varepsilon, Q)$, defined to be the infimum of $D(P||Q)$
over all probability measures $P$ that are at least $\varepsilon$-far away from $Q$ in total variation, i.e.,
$$D^*(\varepsilon, Q) = \inf_{P \colon |P-Q| \geq \varepsilon} D(P \| Q), \quad \varepsilon \in [0,2].$$
Note that $D(P\|Q) < \infty$ yields that $\text{Supp}(P) \subseteq \text{Supp}(Q)$.
From Sanov's theorem (see \cite[Theorem~11.4.1]{Cover_Thomas}), $D^*(\varepsilon, Q)$ is equal to the
asymptotic exponential decay of the probability that the total variation distance between the empirical
distribution of a sequence of i.i.d. random variables and the true distribution $(Q)$ is more than a
specified value $\varepsilon$. Upper and lower bounds on $D^*(\varepsilon, Q)$
have been introduced in \cite[Theorem~1]{BerendHK_IT14}, in terms of the balance coefficient
$\beta \geq \frac{1}{2}$ that is defined as
$$\beta \triangleq \inf \left\{ x \in \bigl\{Q(A) \colon A \in \mathcal{F} \bigr\} \colon x \geq \frac{1}{2} \right\}.$$
It has been demonstrated in \cite[Theorem~1]{BerendHK_IT14} that
\begin{align}
D^*(\varepsilon, Q) = C \varepsilon^2 + O(\varepsilon^3)
\label{eq: BHK14}
\end{align}
where
$$\frac{1}{4(2\beta-1)} \, \log \left(\frac{\beta}{1-\beta}\right) \leq C \leq \frac{\log e}{8\beta(1-\beta)}.$$
If the support of $Q$ is a finite set $\mathcal{A}$,
Theorem~\ref{theorem: upper bound on KL divergence in terms of TV distance for probability measures on a finite set}
and \eqref{eq: lower bound on beta_2} yield that
\begin{align*}
D^*(\varepsilon, Q) \leq \log \left(1 + \frac{\varepsilon^2}{2 Q_{\min}}\right)
- \frac{\log e}{2} \cdot \frac{1}{Q_{\max}} \cdot \bigl( Q_{\min} - \varepsilon \bigr)^+.
\end{align*}
Hence, it follows that $D^*(\varepsilon, Q) \leq C_1 \varepsilon^2 + O(\varepsilon^3)$ where
\begin{equation*}
C_1 = \frac{\log e}{2} \left(\frac{1}{Q_{\min}} - \frac{Q_{\min}}{Q_{\max}}\right).
\end{equation*}
Similarly to \eqref{eq: BHK14}, the same quadratic scaling of $D^*(\varepsilon, Q)$ holds for small values of $\varepsilon$,
but with different coefficients.
\end{remark}

\subsection{Example: Total Variation Distance From the Equiprobable Distribution}
Let $\mathcal{A}$ be a finite set, and let $U$ be the equiprobable probability measure on $\mathcal{A}$
(i.e., $U(a) = \frac{1}{|\mathcal{A}|}$ for every $a \in \mathcal{A}$).
The relative entropy of an arbitrary distribution $P$ on $\mathcal{A}$ with respect to the equiprobable
distribution satisfies
\begin{align}
D(P \| U) = \log | \mathcal{A} | - H(P).
\label{eq: relative entropy w.r.t. the equiprobable distribution}
\end{align}
From Pinsker's inequality \eqref{eq: Pinsker's inequality}, the following upper bound on the total variation
distance holds:
\begin{align}
|P-U| \leq \sqrt{ \frac{2}{\log e} \cdot \bigl( \log | \mathcal{A} | - H(P) \bigr) }.
\label{eq: 1st upper bound on TV distance from P to the uniform distribution}
\end{align}
From \cite[Theorem~2.51]{Verdu_book}, for all probability measures $P$ and $Q$,
$$ |P-Q| \leq 2 \sqrt{1 - \exp\bigl( -D(P\|Q) \bigr)} $$
which gives the second upper bound
\begin{align}
|P-U| \leq 2 \sqrt{1 - \frac{1}{|\mathcal{A}|} \cdot \exp\bigl( H(P) \bigr)}.
\label{eq: 2nd upper bound on TV distance from P to the uniform distribution}
\end{align}
From Theorem~\ref{theorem: upper bound on KL divergence in terms of TV distance for probability measures on a finite set}
and \eqref{eq: lower bound on beta_2}, we have
\begin{align*}
D(P \| U) \leq \log \left( 1 + \frac{|\mathcal{A}|}{2} \cdot |P-U|^2 \right)
- \left( \frac{|\mathcal{A}| \, \log e}{2} \right) \cdot \left(\frac{1}{|\mathcal{A}|} - |P-U|\right)^+ \cdot |P-U|^2.
\end{align*}
A loosening of the latter bound by removing its second non-negative term on the RHS of this inequality,
in conjunction with \eqref{eq: relative entropy w.r.t. the equiprobable distribution},
leads to the following closed-form expression for the lower bound on the total variation distance:
\begin{align}
|P-U| \geq \sqrt{ 2 \left( \exp\bigl( -H(P) \bigr) - \frac{1}{|\mathcal{A}|} \right) }.
\label{eq: lower bound on TV distance from P to the uniform distribution}
\end{align}
Let $H(P) = \beta \, \log |\mathcal{A}|$, so $\beta \in [0,1]$.
From \eqref{eq: 1st upper bound on TV distance from P to the uniform distribution},
\eqref{eq: 2nd upper bound on TV distance from P to the uniform distribution}
and \eqref{eq: lower bound on TV distance from P to the uniform distribution},
it follows that
\begin{equation}
\sqrt{2 \left[ \left(\frac{1}{|\mathcal{A}|} \right)^{\beta} - \frac{1}{|\mathcal{A}|} \right]}
\leq |P-U| \leq \min \left\{ \sqrt{2 (1-\beta) \ln |\mathcal{A}|}, \;
2 \sqrt{1-|\mathcal{A}|^{\beta-1}} \right\}.
\label{eq: upper and lower bounds on TV distance from P to the uniform distribution}
\end{equation}
As expected, if $\beta=1$, both upper and lower bounds are equal to zero (since $D(P \| U) = 0$).
The lower bound on the LHS of \eqref{eq: upper and lower bounds on TV distance from P to the uniform distribution}
improves the lower bound on the total variation distance which follows from
\eqref{eq: looser bound by Csiszar and Talata}:
\begin{align}
|P-U| \geq \sqrt{\frac{(1-\beta) \ln |\mathcal{A}|}{|\mathcal{A}|}}
\label{eq: looser lower bound on TV distance from the equiprobable distribution}
\end{align}
For example, for a set of size $|\mathcal{A}|=1024$ and $\beta = 0.5$, the improvement
in the new lower bound on the total variation distance is from 0.0582 to 0.2461.

Note that if $\beta \rightarrow 0$ (i.e., $P$ is far in relative entropy from the
equiprobable distribution), and the set $\mathcal{A}$ stays fixed, the ratio between
the upper and lower bounds in
\eqref{eq: upper and lower bounds on TV distance from P to the uniform distribution}
tends to $\sqrt{2}$. On the other hand, in this case, the ratio between the upper
and the looser lower bound in
\eqref{eq: looser lower bound on TV distance from the equiprobable distribution}
tends to $$2\sqrt{\frac{|\mathcal{A}|-1}{\ln |\mathcal{A}|}},$$
which can be made arbitrarily large for a sufficiently large set~$\mathcal{A}$.

\section{Extension of
Theorem~\ref{theorem: upper bound on KL divergence in terms of TV distance for probability measures on a finite set}
to R\'{e}nyi Divergences}
\label{section: Extension to Renyi divergences}
The present section extends
Theorem~\ref{theorem: upper bound on KL divergence in terms of TV distance for probability measures on a finite set}
to R\'{e}nyi divergences of an arbitrary order $\alpha \in [0, \infty]$ (i.e., it relies on
Theorem~\ref{theorem: upper bound on KL divergence in terms of TV distance for probability measures on a finite set}
to provide a generalization of the special case where $\alpha = 1$), and the use of this generalized inequality
is exemplified.

\subsection{Main Result and Proof}
The following theorem provides a kind of a generalized reverse Pinsker inequality
where the R\'{e}nyi divergence of an arbitrary order $\alpha \in [0,\infty]$
is upper bounded in terms of the total variation distance for probability
measures defined on a common finite set.

\vspace*{0.1cm}
\begin{theorem}
Let $P$ and $Q$ be probability measures on a common finite set $\mathcal{A}$, and
assume that $P, Q$ are strictly positive.
Let $\varepsilon \triangleq |P-Q|$ (recall that $\varepsilon \in [0,2]$),
$\varepsilon' \triangleq \min\{1, \varepsilon\}$, and
$$P_{\min} \triangleq \min_{a \in \mathcal{A}} P(a), \quad Q_{\min} \triangleq \min_{a \in \mathcal{A}} Q(a).$$
Then, the R\'{e}nyi divergence of order $\alpha \in [0,\infty]$ satisfies
\begin{align}
D_{\alpha}(P \| Q) \leq \left\{
\begin{array}{ll}
\log \left(1 + \frac{\varepsilon}{2Q_{\min}} \right), & \mbox{if $\alpha \in (2, \infty]$} \\[0.3cm]
\log \left(1 + \frac{\varepsilon \, \varepsilon'}{2Q_{\min}} \right), & \mbox{if $\alpha \in [1,2]$} \\[0.3cm]
\min \left\{ f_1, f_2 \right\},
& \mbox{if $\alpha \in \bigl(\frac{1}{2}, 1 \bigr)$} \\[0.4cm]
\min \left\{ -2 \log\left(1-\frac{\varepsilon}{2} \right), f_1, f_2 \right\},
& \mbox{if $\alpha \in \bigl[0, \frac{1}{2}\bigr]$}
\end{array}
\right.
\label{eq: upper bound on Renyi divergence in terms of TV distance for probability measures on a finite set}
\end{align}
where, for $\alpha \in [0,1)$,
\begin{align}
& f_1 \equiv f_1(\alpha, P_{\min}, Q_{\min}, \varepsilon) \triangleq
\left(\frac{\alpha}{1-\alpha}\right)
\left[ \log \left(1 + \frac{\varepsilon^2}{2P_{\min}} \right) - \left(\frac{Q_{\min} \, \log e}{2}\right)
\varepsilon^2 \right], \label{eq: f1} \\[0.2cm]
& f_2 \equiv f_2(P_{\min}, Q_{\min}, \varepsilon, \varepsilon') \triangleq \log \left(1 + \frac{\varepsilon \, \varepsilon'}{2Q_{\min}} \right)
- \left(\frac{P_{\min} \, \log e}{2}\right) \varepsilon^2. \label{eq: f2}
\end{align}
\label{theorem: upper bound on Renyi divergence in terms of TV distance for probability measures on a finite set}
\end{theorem}

\begin{proof}
The R\'{e}nyi divergence of order $\infty$ satisfies (see, e.g., \cite[Theorem~6]{ErvenH14})
$$D_{\infty}(P||Q) = \log \left(\text{ess sup} \, \frac{P}{Q}\right).$$
Since, by assumption, the probability measures $P$ and $Q$ are defined
on a common finite set~$\mathcal{A}$
\begin{align}
D_{\infty}(P||Q) & = \log \left(\max_{a \in \mathcal{A}} \frac{P(a)}{Q(a)}\right) \nonumber \\[0.2cm]
& = \log \left(1 + \max_{a \in \mathcal{A}} \frac{P(a)-Q(a)}{Q(a)}\right) \nonumber \\[0.2cm]
& \leq \log \left(1 + \frac{\max_{a \in \mathcal{A}} |P(a)-Q(a)|}{\min_{a \in \mathcal{A}} Q(a)} \right) \nonumber \\[0.2cm]
& \leq \log \left(1 + \frac{|P-Q|}{2Q_{\min}} \right)
\label{eq: upper bound on Renyi divergence of order infinity}
\end{align}
where the last inequality follows from
\eqref{eq: second intermediate step for the derivation of the upper bound on the chi-square divergence}.
Since the R\'{e}nyi divergence of order $\alpha \in [0, \infty]$
is monotonic non-decreasing in~$\alpha$ (see, e.g., \cite[Theorem~3]{ErvenH14}),
it follows from \eqref{eq: upper bound on Renyi divergence of order infinity} that
\begin{align}
D_{\alpha}(P \| Q) \leq D_{\infty}(P \| Q) \leq \log \left(1 + \frac{\varepsilon}{2Q_{\min}} \right),
\quad \forall \, \alpha \in [0, \infty]
\label{eq: 1st bound for Renyi divergence}
\end{align}
which proves the first line in
\eqref{eq: upper bound on Renyi divergence in terms of TV distance for probability measures on a finite set}
when the validity of the bound is restricted to $\alpha \in (2, \infty]$.

For proving the second line in
\eqref{eq: upper bound on Renyi divergence in terms of TV distance for probability measures on a finite set},
it is shown that the bound in
\eqref{eq: upper bound on KL divergence in terms of TV distance for probability measures on a finite set}
can be sharpened by replacing $D(P \| Q)$ on the LHS of
\eqref{eq: upper bound on KL divergence in terms of TV distance for probability measures on a finite set}
with the quadratic R\'{e}nyi divergence $D_2(P \| Q)$ (note that $D_2(P \| Q) \geq D(P \| Q)$), leading to
\begin{equation}
D_2(P \| Q) \leq \log \left(1 + \frac{|P-Q|^2}{2 Q_{\min}} \right).
\label{eq: upper bound on quadratic Renyi divergence in terms of TV distance for probability measures on a finite set}
\end{equation}
The strengthened inequality in
\eqref{eq: upper bound on quadratic Renyi divergence in terms of TV distance for probability measures on a finite set},
in comparison to \eqref{eq: upper bound on KL divergence in terms of TV distance for probability measures on a finite set},
follows by replacing inequality \eqref{eq: lower bound on the chi-square divergence in terms of the KL divergence}
with the equality in \eqref{eq: relation between the chi-square and quadratic Renyi divergences}. Combining
\eqref{eq: upper bound on the chi-square divergence} and
\eqref{eq: relation between the chi-square and quadratic Renyi divergences}
gives inequality
\eqref{eq: upper bound on quadratic Renyi divergence in terms of TV distance for probability measures on a finite set}, and
\begin{equation}
D_{\alpha}(P\|Q) \leq D_2(P \| Q) \leq \log \left(1 + \frac{\varepsilon^2}{2 Q_{\min}} \right), \quad
\forall \, \alpha \in [0,2].
\label{eq: 2nd bound for Renyi divergence}
\end{equation}
The combination of \eqref{eq: 1st bound for Renyi divergence}
with \eqref{eq: 2nd bound for Renyi divergence} gives the second line in
\eqref{eq: upper bound on Renyi divergence in terms of TV distance for probability measures on a finite set}
(note that $\varepsilon \varepsilon' = \min \{\varepsilon, \varepsilon^2\}$) while the validity of the bound
is restricted to $\alpha \in [1,2]$.

For $\alpha \in (0,1)$, $D_{\alpha}(P\|Q)$ satisfies the
skew-symmetry property $D_{\alpha}(P \| Q) = \frac{\alpha}{1-\alpha} \cdot D_{1-\alpha}(Q \| P)$
(see, e.g., \cite[Proposition~2]{ErvenH14}).
Consequently, we have
\begin{align}
D_{\alpha}(P \| Q) & = \left(\frac{\alpha}{1-\alpha}\right) D_{1-\alpha}(Q \| P) \nonumber \\[0.1cm]
& \leq \left(\frac{\alpha}{1-\alpha}\right) D(Q \| P) \nonumber \\[0.1cm]
& \leq \left(\frac{\alpha}{1-\alpha}\right)
\left[ \log \left(1 + \frac{\varepsilon^2}{2P_{\min}} \right) - \left(\frac{Q_{\min} \, \log e}{2}\right) \varepsilon^2 \right],
\quad \forall \, \alpha \in (0,1)
\label{eq: 3rd bound for Renyi divergence}
\end{align}
where the first inequality holds since the R\'{e}nyi divergence is monotonic non-decreasing in
its order, and the second inequality follows from
Theorem~\ref{theorem: upper bound on KL divergence in terms of TV distance for probability measures on a finite set}
which implies that
\begin{align*}
D(Q\|P) & \leq \log \left(1 + \frac{\varepsilon^2}{2P_{\min}} \right)
- \frac{\log e}{2} \cdot \min_{a \in \mathcal{A}} \frac{Q(a)}{P(a)} \cdot \varepsilon^2 \\
& \leq \log \left(1 + \frac{\varepsilon^2}{2P_{\min}} \right)
- \left(\frac{Q_{\min} \, \log e}{2} \right) \varepsilon^2.
\end{align*}
The third line in
\eqref{eq: upper bound on Renyi divergence in terms of TV distance for probability measures on a finite set}
follows from \eqref{eq: 1st bound for Renyi divergence}, \eqref{eq: 2nd bound for Renyi divergence}
and \eqref{eq: 3rd bound for Renyi divergence} while restricting the validity of the bound
to $\alpha \in \bigl(\frac{1}{2}, 1\bigr)$.

For proving the fourth line in
\eqref{eq: upper bound on Renyi divergence in terms of TV distance for probability measures on a finite set},
note that from \eqref{eq: Renyi divergence}
$D_{1/2}(P \| Q) = -2 \log Z(P,Q)$
where $Z(P,Q) \triangleq \sum_{a \in \mathcal{A}} \sqrt{ P(a) Q(a) }$ is the Bhattacharyya coefficient
between $P$ and $Q$ \cite{Kailath67}. The Bhattacharyya distance is defined as minus the logarithm
of the Bhattacharyya coefficient, which is non-negative in general and it is zero if and only if $P=Q$
(since $0 \leq Z(P,Q) \leq 1$, and $Z(P,Q)=1$ if and only if $P=Q$). Hence, the R\'{e}nyi divergence of
order $\frac{1}{2}$ is twice the Bhattacharyya distance. Based on the inequality
$Z(P,Q) \geq 1 - \frac{|P-Q|}{2}$, which follows from \cite[Example~6.2]{GSS_IT14} (see also
\cite[Proposition~1]{Sason_ITW15_Jerusalem}), we have
\begin{equation}
D_{\alpha}(P \| Q) \leq D_{1/2}(P \| Q) \leq -2 \log\left(1-\frac{\varepsilon}{2} \right), \quad \forall \,
\alpha \in \Bigl[0, \frac{1}{2} \Bigr]
\label{eq: 4th bound for Renyi divergence}
\end{equation}
where $\varepsilon \triangleq |P-Q| \in [0,2]$.
Finally, the last case in
\eqref{eq: upper bound on Renyi divergence in terms of TV distance for probability measures on a finite set}
follows from \eqref{eq: 1st bound for Renyi divergence}, \eqref{eq: 2nd bound for Renyi divergence},
\eqref{eq: 3rd bound for Renyi divergence} and~\eqref{eq: 4th bound for Renyi divergence}.
\end{proof}

\subsection{Example: R\'{e}nyi Divergence for Multinomial Distributions}

Let $X_1, X_2, \ldots$ be independent Bernoulli random variables with $X_i \sim \text{Bernoulli}(p_i)$,
and let $Y_1, Y_2, \ldots$ be independent Bernoulli random variables with $Y_i \sim \text{Bernoulli}(q_i)$
(assume w.l.o.g. that $q_i \leq \frac{1}{2}$).
Let $U_n$ and $V_n$ be the partial sums $U_n = \sum_{i=1}^n X_i$ and $V_n = \sum_{i=1}^n Y_i$, and let
$P_{U_n}, P_{V_n}$ denote their multinomial distributions. For all $\alpha \in [0,2]$ and $n \in \naturals$,
we have
\begin{align}
& D_{\alpha}(P_{U_n} \| P_{V_n}) \nonumber \\
& \stackrel{\text{(a)}}{\leq} D_{\alpha}(P_{X_1, \ldots, X_n} \| P_{Y_1, \ldots, Y_n}) \nonumber \\
& \stackrel{\text{(b)}}{=} \sum_{i=1}^n D_{\alpha}(P_{X_i} \| P_{Y_i}) \nonumber \\
& \stackrel{\text{(c)}}{\leq} \log \left(1 + \frac{|P_{X_i} - P_{Y_i}|^2}{2 \, \bigl(P_{Y_i}\bigr)_{\min}} \right) \nonumber \\
& \stackrel{\text{(d)}}{=} \sum_{i=1}^n \log \left(1 + 2q_i \left(\frac{p_i}{q_i}-1\right)^2 \right)
\label{eq: example - bound on Renyi divergence with alpha below 2}
\end{align}
where inequality~(a) follows from the data processing inequality for the R\'{e}nyi divergence (see
\cite[Theorem~9]{ErvenH14}), equality~(b) follows from the additivity property of the R\'{e}nyi divergence
under the independence assumption for $\{X_i\}$ and for $\{Y_i\}$ (see \cite[Theorem~28]{ErvenH14}),
inequality~(c) follows from
Theorem~\ref{theorem: upper bound on Renyi divergence in terms of TV distance for probability measures on a finite set},
and equality~(d) holds since $|P_{X_i} - P_{Y_i}| = 2 |p_i-q_i|$ for Bernoulli random variables, and
$(P_{Y_i})_{\min} = \min\{q_i, 1-q_i\} = q_i$ $(q_i \leq \frac{1}{2})$. Similarly, for all $\alpha > 2$
and $n \in \naturals$,
\begin{equation}
D_{\alpha}(P_{U_n} \| P_{V_n}) \leq \sum_{i=1}^n \log \left(1 + 2 \, \left|\frac{p_i}{q_i}-1\right| \right).
\label{eq: example - bound on Renyi divergence with alpha above 2}
\end{equation}
The only difference in the derivation of \eqref{eq: example - bound on Renyi divergence with alpha above 2}
is in inequality~(c) of \eqref{eq: example - bound on Renyi divergence with alpha below 2}
where the bound in the first line of
\eqref{eq: upper bound on Renyi divergence in terms of TV distance for probability measures on a finite set}
is used this time.

Let $\{\varepsilon_n\}_{n=1}^{\infty}$ be a non-negative sequence such that
$$(1-\varepsilon_n) q_n \leq p_n \leq (1+\varepsilon_n) q_n, \quad \forall \, n \in \naturals$$
and $$\sum_{n=1}^{\infty} \varepsilon_n^2 < \infty.$$
Then, from \eqref{eq: example - bound on Renyi divergence with alpha below 2}, it follows
that $D_{\alpha}(P_{U_n} \| P_{V_n}) \leq K_1$ for all $\alpha \in [0,2]$ and $n \in \naturals$ where
$$K_1 \triangleq \sum_{n=1}^{\infty} \log \left( 1 + \varepsilon_n^2 \right) < \infty.$$

Furthermore, if $\sum_{n=1}^{\infty} \varepsilon_n < \infty$, it follows from
\eqref{eq: example - bound on Renyi divergence with alpha above 2} that
$D_{\alpha}(P_{U_n} \| P_{V_n}) \leq K_2$ for all $\alpha > 2$ and $n \in \naturals$ where
$$K_2 \triangleq \sum_{n=1}^{\infty} \log \left( 1 + 2 \varepsilon_n \right) < \infty.$$

Note that although $D_{\alpha}(P_{X_i} \| P_{Y_i})$ in equality~(b)
of \eqref{eq: example - bound on Renyi divergence with alpha below 2} is equal to the
binary R\'{e}nyi divergence
\begin{align*}
d_{\alpha}(p_i \| q_i) \triangleq
\left\{ \begin{array}{ll}
\left(\frac{1}{\alpha-1} \right)
\log \Bigl(p_i^\alpha q_i^{1-\alpha}+(1-p_i)^\alpha (1-q_i)^{1-\alpha} \Bigr),
& \mbox{if $\alpha \in (0,1) \cup (1, \infty)$} \\[0.2cm]
p_i \, \log\left(\frac{p_i}{q_i}\right) + (1-p_i) \, \log\left(\frac{1-p_i}{1-q_i}\right),
& \mbox{if $\alpha = 1$}
\end{array}
\right.
\end{align*}
the reason for the use of the upper bounds in step~(c) of
\eqref{eq: example - bound on Renyi divergence with alpha below 2} and
\eqref{eq: example - bound on Renyi divergence with alpha above 2} is to
state sufficient conditions, in terms of $\{\varepsilon_n\}_{n=1}^{\infty}$,
for the boundedness of the R\'{e}nyi divergence $D_{\alpha}(P_{U_n} \| P_{V_n})$.

\section{The Exponential Decay of the Probability for a Non-Typical Sequence}
\label{section: The Exponential Decay of the Probability for a Non-Typical Sequence}

Let $U^N = (U_1, \ldots, U_N)$ be a sequence of i.i.d. symbols that are emitted by a
memoryless and stationary source with distribution $Q$ and a finite alphabet $\mathcal{A}$.
Let $|\mathcal{A}| = r < \infty$ denote the cardinality of the source alphabet, and assume
that all symbols are emitted with positive probability (i.e.,
$Q_{\min} \triangleq \min_{a \in \mathcal{A}} Q(a) > 0$). The empirical probability distribution
of the emitted sequence $\hat{P}_{U^N}$ is given by
\begin{align*}
\hat{P}_{U^N}(a) \triangleq \frac{1}{N} \sum_{k=1}^N 1\{U_k=a\}, \quad \forall \, a \in \mathcal{A}.
\end{align*}
For an arbitrary $\delta>0$, let the $\delta$-typical set be defined as
\begin{align}
T_{Q}(\delta) \triangleq \left\{ u^N \in \mathcal{A}^N \colon \bigl| \hat{P}_{u^N}(a) - Q(a) \bigr|
< \delta \, Q(a), \quad \forall \, a \in \mathcal{A} \right\},
\label{eq: delta typical set}
\end{align}
i.e., the empirical distribution of every symbol in an $N$-length $\delta$-typical sequence
deviates from the true distribution of this symbol by a fraction of less than $\delta$.
Consequently, the complementary of \eqref{eq: delta typical set} is given by
$$T_{Q}(\delta)^{\text{c}} = \left\{ u^N \in \mathcal{A}^N \colon \exists \, a \in \mathcal{A}, \; \;
\bigl| \hat{P}_{u^N}(a) - Q(a) \bigr| \geq \delta \, Q(a) \right\}.$$
From Sanov's theorem (see \cite[Theorem~11.4.1]{Cover_Thomas}), the asymptotic exponential
decay of the probability that a sequence $U^N$ is not $\delta$-typical, for a
specified $\delta>0$, is given by
\begin{align}
\lim_{N \rightarrow \infty} -\frac{1}{N} \, \log Q^N\bigl(T_{Q}(\delta)^{\text{c}}\bigr) = \min_{P \in \mathcal{P}_Q} D(P\|Q)
\label{eq: exponential decay rate of the probability for non-typicality}
\end{align}
where
\begin{align}
\mathcal{P}_Q \triangleq \Bigl\{P \; \text{is a probability measure on} \, (\mathcal{A}, \mathcal{F})
\colon \exists \, a \in \mathcal{A}, \; \; |P(a)-Q(a)| \geq \delta \, Q(a) \Bigr\}.
\label{eq: set of non-typical probability distributions}
\end{align}

We obtain in the following explicit upper and lower bounds on the exponential decay rate
on the RHS of \eqref{eq: exponential decay rate of the probability for non-typicality}. The emphasis
is on the upper bound, which is based on
Theorem~\ref{theorem: upper bound on KL divergence in terms of TV distance for probability measures on a finite set},
and we first introduce the lower bound for completeness.
The derivation of the lower bound is similar to the analysis in \cite[Section~4]{OrdentlichW_IT2005}; note, however,
that there is a difference between the $\delta$-typicality
in \cite[Eq.~(19)]{OrdentlichW_IT2005} and the way it is defined in \eqref{eq: delta typical set}.
The probability-dependent refinement of Pinsker's inequality (see \cite[Theorem~2.1]{OrdentlichW_IT2005})
states that
\begin{align}
D(P \| Q) \geq \varphi(\pi_Q) \; |P-Q|^2
\label{eq: probability-dependent refinement of Pinsker's inequality}
\end{align}
where
\begin{align}
\pi_Q \triangleq \max_{A \in \mathcal{F}} \min \bigl\{Q(A), 1-Q(A)\bigr\} \leq \frac{1}{2}
\label{eq: pi_Q}
\end{align}
and
\begin{align}
\varphi(p) = \left\{
\begin{array}{cl}
\frac{1}{4(1-2p)} \, \log \left( \frac{1-p}{p} \right), & \quad \mbox{if $p \in
\bigl[0, \frac{1}{2} \bigr)$,} \\[0.2cm]
\frac{\log e}{2}, & \quad
\mbox{if $p=\frac{1}{2}$}
\end{array}
\right.
\label{eq: phi function of the probability-dependent refinement of Pinsker's inequality}
\end{align}
is a monotonic decreasing and continuous function. Hence, $\varphi(\pi_Q) \geq \frac{\log e}{2}$,
and \eqref{eq: probability-dependent refinement of Pinsker's inequality} forms a probability-dependent
refinement of Pinsker's inequality \cite{OrdentlichW_IT2005}.
From \eqref{eq: set of non-typical probability distributions} and
\eqref{eq: probability-dependent refinement of Pinsker's inequality}, we have
\begin{align}
& \min_{P \in \mathcal{P}_Q} D(P\|Q) \nonumber \\[0.1cm]
& \geq \varphi(\pi_Q) \; \min_{P \in \mathcal{P}_Q} |P-Q|^2 \nonumber \\
& = \varphi(\pi_Q) \, \left( \min_{a \in \mathcal{A}} \delta \, Q(a) \right)^2 \nonumber \\[0.1cm]
& = \varphi(\pi_Q) \, Q_{\min}^2 \, \delta^2 \triangleq E_{\text{L}} \label{eq: refined lower bound on the exponent} \\
& \geq \left(\frac{Q_{\min}^2 \; \log e}{2} \right) \, \delta^2 \label{eq: looser lower bound on the exponent}
\end{align}
where the transition from \eqref{eq: refined lower bound on the exponent} to \eqref{eq: looser lower bound on the exponent}
follows from the global lower bound on $\varphi(\pi_Q)$.

We derive in the following an upper bound on the asymptotic exponential decay rate
in \eqref{eq: exponential decay rate of the probability for non-typicality}:
\begin{align}
& \min_{P \in \mathcal{P}_Q} D(P\|Q) \nonumber \\[0.1cm]
& \stackrel{\text{(a)}}{\leq} \min_{P \in \mathcal{P}_Q} \left\{ \log \left(1 + \frac{|P-Q|^2}{2 Q_{\min}} \right) \right\} \nonumber \\[0.1cm]
& = \log \left(1 + \frac{ \bigl(\min_{P \in \mathcal{P}_Q} |P-Q|\bigr)^2}{2 Q_{\min}} \right) \nonumber\\[0.1cm]
& \stackrel{\text{(b)}}{=} \log \left(1 + \frac{\bigl(\min_{a \in \mathcal{A}} \, (\delta \, Q(a) \bigr)^2}{2 Q_{\min}} \right) \nonumber\\[0.1cm]
& = \log \left(1 + \frac{Q_{\min} \, \delta^2}{2} \right) \triangleq E_{\text{U}}
\label{eq: upper bound on the exponent}
\end{align}
where inequality~(a) follows from \eqref{eq: upper bound on KL divergence in terms of TV distance for probability measures on a finite set},
and equality~(b) follows from \eqref{eq: set of non-typical probability distributions}.

The ratio between the upper and lower bounds on the asymptotic exponent in
\eqref{eq: exponential decay rate of the probability for non-typicality}, as given
in \eqref{eq: refined lower bound on the exponent} and \eqref{eq: upper bound on the exponent}
respectively, satisfies
\begin{align}
1 & \leq \frac{E_{\text{U}}}{E_{\text{L}}} \nonumber \\[0.1cm]
&= \frac{1}{Q_{\min}} \cdot \frac{\log e}{2 \, \varphi(\pi_Q)} \cdot
\frac{\log \left(1 + \frac{Q_{\min} \, \delta^2}{2} \right)}{\frac{\log e
\; \cdot \; Q_{\min} \, \delta^2}{2} } \label{eq: ratio of bounds on the exponent} \\[0.1cm]
& \leq \frac{1}{Q_{\min}} \nonumber
\end{align}
where inequality~\eqref{eq: ratio of bounds on the exponent} follows from the fact that
the second and third multiplicands in \eqref{eq: ratio of bounds on the exponent} are
both less than or equal to~1.
Note that both bounds in \eqref{eq: refined lower bound on the exponent} and
\eqref{eq: upper bound on the exponent} scale like $\delta^2$ for $\delta \approx 0$.

\section*{Appendix: A Proof of Inequality \eqref{eq: relation between relative entropy, dual of relative entropy and chi-squared divergence}}
This appendix proves inequality \eqref{eq: relation between relative entropy, dual of relative entropy and chi-squared divergence},
which provides upper and lower bounds on the difference $\log\bigl(1+\chi^2(P,Q)\bigr) - D(P\|Q)$ in terms
of the dual relative entropy $D(Q\|P)$. To this end, we first prove a new inequality relating $f$-divergences
\cite{Sason_ITW15_Jerusalem}, and the bounds in
\eqref{eq: relation between relative entropy, dual of relative entropy and chi-squared divergence}
then follow as a special case.

Recall the following definition of an $f$-divergence:
\begin{definition}
Let $f \colon (0, \infty) \rightarrow \reals$ be a convex function with $f(1)=0$,
and let $P$ and $Q$ be two probability measures defined on a common finite set $\mathcal{A}$.
The {\em $f$-divergence} from $P$ to $Q$ is defined by
\begin{equation}
D_f(P||Q) \triangleq \sum_{a \in \mathcal{A}} Q(a) \, f\left(\frac{P(a)}{Q(a)}\right)
\label{eq:f-divergence}
\end{equation}
with the convention that
\begin{align}
&\ 0 f\Bigl(\frac{0}{0}\Bigr) = 0, \quad
f(0) = \lim_{t \rightarrow 0^+} f(t), \nonumber \\
&\ 0 f\Bigl(\frac{b}{0}\Bigr) = \lim_{t \rightarrow 0^+}
t f\Bigl(\frac{b}{t}\Bigr) = b \lim_{u \rightarrow \infty} \frac{f(u)}{u}, \quad \forall \, b > 0.
\label{eq: conventions}
\end{align}
\label{definition:f-divergence}
\end{definition}

\vspace*{0.1cm}
\begin{proposition}
Let $f \colon (0, \infty) \rightarrow \reals$ be a convex function with $f(1)=0$ and
assume that the function $g \colon (0, \infty) \rightarrow \reals$, defined by
$g(t)=-t f(t)$ for every $t>0$, is also convex. Let $P$ and $Q$ be two probability
measures that are defined on a finite set $\mathcal{A}$, and assume that $P, Q$ are
strictly positive. Then, the following inequality holds:
\begin{align}
\min_{a \in \mathcal{A}} \frac{P(a)}{Q(a)} \cdot D_f(P||Q)
\leq -D_g(P||Q) - f\bigl(1 + \chi^2(P,Q)\bigr)
\leq \max_{a \in \mathcal{A}} \frac{P(a)}{Q(a)} \cdot D_f(P||Q).
\label{eq:new inequality relating f-divergences}
\end{align}
\label{proposition:new inequality relating f-divergences}
\end{proposition}

\begin{proof}
Let $\mathcal{A} = \bigl\{a_1, \ldots, a_n \bigr\}$, and
$\underline{u} = (u_1, \ldots, u_n) \in \reals_+^n$ be an arbitrary $n$-tuple
with positive entries. Define
\begin{align}
\begin{split}
& J_n(f, \underline{u}, P) \triangleq \sum_{i=1}^n P(a_i) \, f(u_i)
- f\left(\sum_{i=1}^n P(a_i) \, u_i \right), \\[0.1cm]
& J_n(f, \underline{u}, Q) \triangleq \sum_{i=1}^n Q(a_i) \, f(u_i)
- f\left(\sum_{i=1}^n Q(a_i) \, u_i \right).
\label{eq:Jensen functional}
\end{split}
\end{align}
The following refinement of Jensen's inequality has been introduced in
\cite[Theorem~1]{Dragomir06} for a convex function $f \colon (0, \infty) \rightarrow \reals$:
\begin{align}
\min_{i \in \{1, \ldots, n\}} \frac{P(a_i)}{Q(a_i)} \cdot J_n(f, \underline{u}, Q)
\leq J_n(f, \underline{u}, P)
\leq \max_{i \in \{1, \ldots, n\}} \frac{P(a_i)}{Q(a_i)} \cdot J_n(f, \underline{u}, Q).
\label{eq:Dragomir's inequality ('06)}
\end{align}
Let $u_i \triangleq \frac{P(a_i)}{Q(a_i)}$ for $i \in \{1, \ldots, n\}$.
Calculation of \eqref{eq:Jensen functional} gives that
\begin{align}
J_n(f, \underline{u}, Q)
& = \sum_{i=1}^n Q(a_i) \, f\left(\frac{P(a_i)}{Q(a_i)}\right)
- f\left(\sum_{i=1}^n Q(a_i) \cdot \frac{P(a_i)}{Q(a_i)} \right) \nonumber \\
& = \sum_{a \in \mathcal{A}} Q(a) \, f\left(\frac{P(a)}{Q(a)}\right) - f(1) \nonumber \\[0.1cm]
& = D_f(P||Q), \label{eq:J1} \\[0.1cm]
J_n(f, \underline{u}, P)
& = \sum_{i=1}^n P(a_i) \, f\left(\frac{P(a_i)}{Q(a_i)}\right)
- f\left(\sum_{i=1}^n \frac{P(a_i)^2}{Q(a_i)} \right) \nonumber \\
& \stackrel{(\text{a})}{=} - \sum_{i=1}^n Q(a_i) \, g\left(\frac{P(a_i)}{Q(a_i)}\right)
- f\left(\sum_{i=1}^n \frac{P(a_i)^2}{Q(a_i)} \right) \nonumber \\
& \stackrel{(\text{b})}{=} - D_g(P||Q) - f\bigl(1+\chi^2(P,Q)\bigr)
\label{eq:J2}
\end{align}
where equality~(a) holds by the definition of $g$, and equality~(b) follows from
equalities \eqref{eq: chi-squared divergence} and \eqref{eq:f-divergence}. The
substitution of \eqref{eq:J1} and \eqref{eq:J2} into \eqref{eq:Dragomir's inequality ('06)}
gives inequality~\eqref{eq:new inequality relating f-divergences}.
\end{proof}

As a consequence of Proposition~\ref{proposition:new inequality relating f-divergences},
we prove inequality~\eqref{eq: relation between relative entropy, dual of relative entropy and chi-squared divergence}.
Let $f(t) = -\log(t)$ for $t>0$. The function $f \colon (0, \infty) \rightarrow \reals$
is convex with $f(1)=0$, and $g(t) = -t f(t) = t \log(t)$ for $t>0$ is also convex with
$g(1)=0$. Inequality~\eqref{eq: relation between relative entropy, dual of relative entropy and chi-squared divergence}
follows by substituting $f, g$ into \eqref{eq:new inequality relating f-divergences}
where $D_f(P||Q) = D(Q||P)$ and $D_g(P||Q) = D(P||Q)$. Inequality
\eqref{eq: relation between relative entropy, dual of relative entropy and chi-squared divergence}
also holds in the case where $P$ is not strictly positive on $\mathcal{A}$
with the convention in \eqref{eq: conventions} where $0 \log 0 = \lim_{t \rightarrow 0^+} g(t) = 0$.

\section*{Acknowledgment}
Sergio Verd\'{u} is gratefully acknowledged for his earlier results in \cite{Verdu_ITA14}
that attracted my interest and motivated this work, for providing a draft of \cite{Verdu_book},
and raising the question that led to the inclusion of
Remark~\ref{remark: extension of the proof to general probability measures}.
Vincent Tan is acknowledged for pointing out \cite{TomamichelT_IT13} and
suggesting a simplified proof of
\eqref{eq: second intermediate step for the derivation of the upper bound on the chi-square divergence}.
A discussion with Georg B\"{o}cherer and Bernhard Geiger on their paper \cite{BochererG_arXiv}
has been stimulating along the writing of this manuscript.

\end{document}